\documentclass[11pt,a4paper]{article}

\usepackage[T1]{fontenc}
\usepackage[utf8]{inputenc}
\usepackage{lmodern} %
\usepackage[a4paper,margin=2.5cm]{geometry} %
\usepackage{amsmath,amssymb,amsthm}
\usepackage{booktabs}
\usepackage{graphicx}
\usepackage{float}
\usepackage{tikz}
\usetikzlibrary{patterns,arrows.meta}
\usepackage[round]{natbib}
\usepackage[hidelinks]{hyperref}

\newtheorem{proposition}{Proposition}
\newtheorem{corollary}{Corollary}
\newtheorem{definition}{Definition}

\newcommand{\rga}{\mathrm{RGA}}
\newcommand{\rge}{\mathrm{RGE}}
\newcommand{\rgr}{\mathrm{RGR}}
\newcommand{\aurga}{\mathrm{AURGA}}
\newcommand{\aurge}{\mathrm{AURGE}}
\newcommand{\aurgr}{\mathrm{AURGR}}
\newcommand{\Sm}{S} %
\newcommand{\thetahat}{\hat\theta}
\newcommand{\Sighat}{\widehat\Sigma} %

\title{Confidence regions for SAFE AI compliance scores}
\author{Anton Sokolov\thanks{Tyche Institute, Tallinn, Estonia.
    \texttt{anton.sokolov@tyche.institute}}
  \and Paolo Giudici \thanks{Department of Economics, University of Pavia, Pavia, Italy.}
  \and Vasily Kolesnikov \thanks{Department of Computer Engineering, University of Pavia, Pavia, Italy.}}
\date{August 2026}

\begin{document}
\maketitle

\begin{abstract}
Artificial Intelligence trustworthiness scores are moving from research dashboards into compliance claims and procurement decisions. A compliance claim is a statistical decision based on the comparison between an estimated compliance metric and a set threshold.
The recently proposed integrated SAFE AI metrics are based on three basic rank-graduation  measures of accuracy, explainability and robustness, expressed on a common footing.  The metrics are integrated
into a single compliance score, so far reported as a point value. 

In this paper we propose to add an uncertainty layer to the SAFE AI compliance score. All component metrics are estimated on a
shared test sample, so their errors covary. We estimate the full covariance of the component vector with a paired bootstrap, and quantify departures from the source-wise independence baseline
suggested by the integrated-metrics framework's variability decomposition. A closed-form identity splits such departure into an uncertainty-weighted effective dimension and an exposure-weighted data-driven correlation, with both depending on the aggregator and the estimated uncertainty of the component metrics.

The application of our proposal to both real and simulated data shows that the assumption of independence can either narrow or widen confidence bands, while the direction and magnitude of the effect depend on the aggregator type, machine learning model, and perturbation family. The resulting confidence region
supports compliance decisions, which can be documented in a formal  uncertainty certificate.

\end{abstract}

\section{Introduction}\label{sec:intro}

The increasing applications of Artificial Intelligence require the development of scoring models that can assess their trustworthiness, and decide whether they are compliant to regulations and codes of practices, such as the European AI Act, comparing the obtained score with a set threshold. Some recent papers deal with this problem, and can be categorised in two main streams.

On one hand, reviews of trust, risk and
security management for agentic systems built on large language models
catalogue the dimensions to measure and the controls that use them
\citep{raza2026trism}.  Some composite trust
indices use a one-dimensional percentile bootstrap on the scalar composite
\citep{ahadian2026haiti} to quantify uncertainty on the measurement. 
While these approaches take uncertainty into account, 
 their output covers the scalar alone and omits the
cross-component covariance and a joint statement about the component vector.
Related methods
provide simultaneous inference for fairness auditing across subgroups
\citep{cherian2023} and propagate uncertainty from composite-indicator weights
and construction choices \citep{saisana2005}. The covariance between
heterogeneous trust metrics estimated on a shared test sample remains
unaddressed.

On the other hand, the rank-graduation metrics of
\citet{AuricchioBernardelliGiudiciToscani2026} (RGX) provide the measurement basis to fill the gap.
Rank Graduation Accuracy (RGA)
generalises rank-based predictive accuracy;
Rank Graduation Robustness (RGR) measures the stability of predictions under
perturbation; Rank Graduation Explainability (RGE) measures how much
predictions depend on declared features.  The
metrics have been integrated into a single compliance score for model-level
comparison \citep{giudici2026integrated}, and the same line carries the
metrics into multi-agent settings in which an agent decides by comparing a
metric to a risk-appetite threshold \citep{babaei2026agentic}. While these approaches are mathematically consistent, they produce an integrated compliance score which does not take uncertainty into account. 

In this paper we aim to improve both streams of research: maintain the mathematical consistency of the SAFE-AI approach, but developing uncertainty measures around it, also taking into account the correlation between the different components.

Confidence intervals for the individual RGX  metrics may be derived from the existing literature. For example,  \cite{giudiciraffinetti-adac} introduces  
 a two-model comparison test for RGA built on dependent U-statistics with a
jackknife variance \citep{hoeffding1948, efronstein1981}, and
the Rank Graduation metrics carries that apparatus across the other metrics
\citep{babaei2025rgb}; additionally, RGA is a close relative of AUROC, whose interval
estimation can be based on  \citet{delong1988}, or on confidence
intervals for general rank statistics \citep{newson2006}.
Alternatively, 
 distribution-free
methods based on conformal
prediction can be employed. These methods wrap a model's output for a
\emph{new instance} in a set with a finite-sample guarantee
\citep{bates2021}; the quantity at issue here is a scalar summary of an
\emph{already-fixed} evaluation, computed on a sample the three metrics
share. What is missing  is the covariance between them: a conformal set for the composite metric would still
leave the cross-metric term unmeasured.

Earlier work \citep{sokolov2026composed} introduced a shared-sample covariance argument, using paired-bootstrap and jackknife estimators, two-link propagation, and an expiring calibration certificate.  In this work, we extend this machinery to the volume-integration compliance score defined in \citet{giudici2026integrated}. We show, from both a theoretical and an empirical viewpoint, that assuming independence between the different metrics can have a material effect on the reported uncertainty. 

We make three main contributions. First, we estimate the covariance of the component vector, leading to a joint confidence region for the component metrics and a confidence interval for the compliance score.

Second, we present a closed-form identity that gives the cost of assuming independence, and  show that the perturbation family can affect the sign of the covariance effect. 

Third, we apply the proposal in two settings: a simulated
run of an executed agentic verifier and the well-known German credit lending data.

Figure~\ref{fig:paper-flow} summarises the organisation of the paper. Section~\ref{sec:metrics} defines the SAFE component metrics, the composed compliance score and the independence baseline. Section~\ref{sec:joint} develops the common theory: paired estimation of the shared-sample covariance, the resulting joint confidence region, and the exact diagnostic for the cost of replacing the covariance matrix by its diagonal. Sections~\ref{sec:german} and~\ref{sec:simulated} then apply that same machinery independently to the German-credit and simulated-agent settings. Section~\ref{sec:discussion} draws the conclusions.

\begin{figure}[t]
\centering
\begin{tikzpicture}[
  x=1cm,
  y=1cm,
  flowbox/.style={
    draw=black!65,
    rounded corners=1.5pt,
    align=center,
    text width=5.0cm,
    minimum height=1.65cm,
    inner sep=4pt,
    font=\small
  },
  context/.style={flowbox,fill=black!3},
  method/.style={flowbox,fill=black!5},
  theory/.style={flowbox,fill=black!10},
  application/.style={flowbox,fill=black!6},
  flow/.style={-{Latex[length=2.1mm]},semithick,draw=black!70},
  branch/.style={semithick,draw=black!70}
]

\node[context] (intro) at (0,6.0) {
  \textbf{Introduction}\\[-1pt]
  \footnotesize Section~\ref{sec:intro}\\
  problem and contributions for\\ uncertainty-aware compliance decisions
};

\node[method] (metrics) at (0,3.8) {
  \textbf{Measurement frame}\\[-1pt]
  \footnotesize Section~\ref{sec:metrics}\\
  SAFE curves, compliance score,\\ independence baseline
};

\node[theory] (theory) at (0,1.6) {
  \textbf{Common uncertainty theory}\\[-1pt]
  \footnotesize Section~\ref{sec:joint}\\
  paired covariance, joint region,\\ independence-cost diagnostic
};

\node[application] (german) at (-2.85,-1.25) {
  \textbf{German-credit data}\\[-1pt]
  \footnotesize Section~\ref{sec:german}\\
  chain and aggregator\\ sensitivity
};

\node[application] (simulated) at (2.85,-1.25) {
  \textbf{Simulated-agent data}\\[-1pt]
  \footnotesize Section~\ref{sec:simulated}\\
  within- and cross-stage\\ sensitivity
};

\draw[flow] (intro.south) -- (metrics.north);
\draw[flow] (metrics.south) -- (theory.north);

\coordinate (fork) at (0,0.0);
\draw[branch] (theory.south) -- (fork);
\draw[flow] (fork) -| (german.north);
\draw[flow] (fork) -| (simulated.north);

\end{tikzpicture}
\caption{Paper structure and argumentative dependencies. Sections~1--2
establish the problem and measurement frame; Section~3 develops the common
uncertainty theory; and Sections~4--5 independently instantiate that theory on
German-credit and simulated-agent data. The parallel branches are applications
of the same machinery, not sequential stages. Arrows denote argumentative
dependencies, not causality.}
\label{fig:paper-flow}
\end{figure}
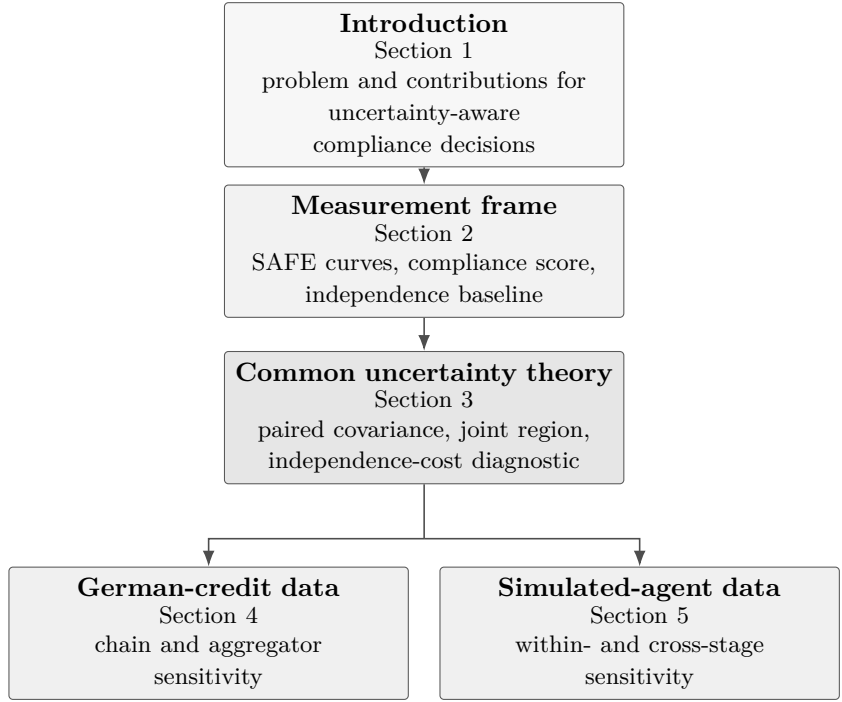

\section{Measurement frame}\label{sec:metrics}

\subsection{The rank-graduation family}

The SAFE AI approach (see, e.g., \cite{giudici2026integrated}) is based on a set of rank graduation metrics built from a concordance-based construction that produces values in $[0,1]$. For a trained model evaluated on a test sample, $\rga$ measures how closely the ranking induced by the predicted values matches the ranking of the observed responses, $\rge$ measures how much the predictions change when features are removed, and $\rgr$ measures how much the predictions change under input data perturbations.  

The three metrics depend on parameters that are typically fixed exogenously according to a risk appetite. To obtain model summaries that do not rely on a single choice, we calculate the areas under the three curves, similarly to what is usually done for the well-known Area Under the ROC Curve. This leads to the Area Under the Accuracy Curve, $\aurga=\int_0^1\rga(p)\,dp$, which summarises accuracy resilience, where $p$ denotes normalised severity along the accuracy curve. The Area Under the Robustness Curve is $\aurgr=\int_0^1\rgr(p)\,dp$, where $p$ indicates normalised perturbation intensity. The Area Under the Explainability Curve is $\aurge=\int_0^1\rge(p)\,dp$, where $p$ is the fraction of features removed.

In empirical implementations, the three curves are evaluated over
increasing levels of severity. The accuracy curve used in this analysis is the binary partial-contribution curve. Observations are ordered by predicted score and divided into segments, whose normalised contributions are cumulatively subtracted from the full-sample RGA. The explainability curve is constructed by greedily removing features. For the robustness curve, the primary perturbation is a tail swap. At severity $p$, observations in the lower and upper $p$-quantiles of each feature are exchanged. The resulting curves are then evaluated on a finite grid of values of the parameter $p$.

\subsection{The compliance score}

The three metrics are expressed in terms of the same underlying concordance framework and can thus be integrated in a single compliance score.
The  term ``compliance score'' follows the source literature and
denotes a statistical composite. %

To derive a compliance score for the SAFE AI metrics \cite{giudici2026integrated} work directly with the three vectors from which the curves are built. Let $t = 1,\ldots,L$ index a common normalised severity grid representing partial RGA contribution removal, feature removal, and perturbation intensity. Writing $a = (a_1, \dots, a_L)$, $e = (e_1, \dots, e_L)$ and $r = (r_1, \dots, r_L)$ for the accuracy, explainability and robustness curve values on the grid, we obtain an $L \times L \times L$ tensor by applying a mean $m$ to each triplet,
\begin{equation}
M(i,j,k) \;=\; m\!\left(a_i, e_j, r_k\right),
\qquad i, j, k \in \{1, \dots, L\}.
\label{eq:tensor}
\end{equation}
The mean can take  alternative forms; for example: the arithmetic mean
$(a_i + e_j + r_k)/3$, the geometric mean $(a_i e_j r_k)^{1/3}$, and the quadratic mean $\sqrt{(a_i^2 + e_j^2 + r_k^2)/3}$. 

A compliance score can then be derived as an equal-weight average
of the tensor values, across all three dimensions:
\begin{equation}
\mathcal V \;=\; \frac{1}{L^3} \sum_{i=1}^{L} \sum_{j=1}^{L} \sum_{k=1}^{L} M(i,j,k).
\label{eq:volume}
\end{equation}

\subsection{The independence baseline}

The integrated-metrics framework associates the three curve constructions with different operations: sampling the evaluation observations, modifying the feature set, and perturbing the inputs. In our implementation, however, all three estimated metrics depend on the same bootstrapped test rows, while $\aurgr$ additionally depends on the perturbation mechanism. 

If each metric's randomness were driven only by its own source, and the sources were independent, the covariance matrix of the estimated vector would be diagonal, and a confidence interval
for the compliance score could be assembled from the three marginal variances. However, the three sources are rarely independent. 
A diagonal covariance is thus a counterfactual baseline, not a default
conclusion; Section~\ref{sec:joint} derives the sampling covariance among the three metrics that it omits. 

The practical cost of independence will also be illustrated by means of a real data application, which will be employed as a running example. 
The application consists of the well-known German credit data, which  can be downloaded as Statlog German credit (OpenML \texttt{credit-g} v1), with 700
training and 300 test rows, $d = 20$ features, stratified split, seed
20260723. 
As machine learning models we will consider a logistic regression and a random forest. We employ $B = 2000$ paired bootstrap
replicates throughout. 

\section{Joint confidence regions}\label{sec:joint}

\subsection{Where dependence comes from}

Let $\thetahat = (\thetahat_{\rga}, \thetahat_{\rge}, \thetahat_{\rgr})$ be the vector of the estimated SAFE-AI metrics of Section~\ref{sec:metrics}, for a fixed trained model, all computed on the same test data $\Sm$, and let $\hat C = g(\thetahat)$ be the estimated compliance score for an aggregator $g$. To derive a confidence interval for $C$ we need to estimate the variance-covariance matrix of $\thetahat$, $\Sighat$. It can be easily shown that the covariance between any two distinct metrics can be decomposed into a within bootstrap sample covariance and a between sample covariance:
\begin{equation}
\operatorname{Cov}(\thetahat_k, \thetahat_l)
\;=\;
\mathbb{E}\!\left[\operatorname{Cov}(\thetahat_k, \thetahat_l \mid \Sm)\right]
\;+\;
\operatorname{Cov}\!\left(\mathbb{E}[\thetahat_k \mid \Sm],\;
\mathbb{E}[\thetahat_l \mid \Sm]\right).
\label{eq:totalcov}
\end{equation}
For a given test sample $S$, the only randomness left is from the
perturbation draw, which is generated independently of the data and of the other two metrics. The first term is thus equal to zero. The second term captures the dependence between different test samples. All three metrics covary across different draws of the test sample, and this induces a cross-metric covariance different from zero. 
Assuming that the three metrics are independent amounts to ignore this second component.

The variability that the data dependency generates differs across the three metrics. The explainability term is
deterministic given the data rows: greedy masking uses a fixed search over the prediction space. Accuracy is
also deterministic given the data. Robustness alone adds perturbation
noise to the variability carried by the data.

\subsection{Estimating the covariance}

We can estimate the covariance matrix, $\Sighat$,  in two ways. The paired bootstrap \citep{efron1979} resamples test rows once per replicate and recomputes the entire vector on the same resample, $B = 2000$ times; the empirical covariance among the $B$ replicate vectors estimates the covariance of $\thetahat$. The jackknife \citep{efronstein1981} leaves one test row out, recomputes the vector, and applies the standard leave-one-out covariance formula. In the empirical applications below, we use the paired-bootstrap covariance estimator.

The additional variability caused by the robustness perturbations may be influential. To control its effect we will compare a \emph{fixed draw} with a \emph{redraw} scheme. In the former case we
draw one perturbation realisation per model and hold it across all
replicates; every resulting quantity is conditional on that realisation.
Under the \emph{redraw} scheme, a fresh perturbation is generated for each original evaluation row in every replicate, seeded by the replicate index, after which the bootstrap indices are applied. The resulting quantities are unconditional with respect to the perturbation randomness under this implementation. 

Table \ref{tab:corr} shows the comparison between the estimated correlations under the two schemes, for the German credit data, and two different machine learning models.

\begin{table}[t]
\caption{Paired-bootstrap cross-metric correlations under the two
perturbation-conditioning schemes, $B = 2000$, German credit at $n = 300$.
The fixed-draw columns are conditional on one perturbation realisation; the
redraw columns are unconditional.}
\label{tab:corr}
\small
\centering
\begin{tabular}{lrrrr}
\toprule
 & \multicolumn{2}{c}{logistic} & \multicolumn{2}{c}{random forest} \\
\cmidrule(lr){2-3}\cmidrule(lr){4-5}
Correlation & fixed draw & redraw & fixed draw & redraw \\
\midrule
$\rga$ with $\rge$ & $+0.286$ & $+0.286$ & $+0.292$ & $+0.292$ \\
$\rga$ with $\rgr$ & $+0.160$ & $+0.127$ & $+0.153$ & $+0.139$ \\
$\rge$ with $\rgr$ & $+0.597$ & $+0.430$ & $+0.342$ & $+0.329$ \\
\bottomrule
\end{tabular}
\end{table}

From Table~\ref{tab:corr}, note that the $\rga$-with-$\rge$ correlations are identical. This holds by construction, since both functionals are independent of the perturbations. Only the pairs involving $\rgr$ vary under redrawing, as expected. The dispersion of $\rgr$ also changes under redrawing, reflecting the integration over the additional perturbation randomness.
In particular,  note that  $\rgr$ dispersion increases under redrawing. This is an empirical finding rather than a direct consequence of the redraw scheme.

The estimated covariance induces an ellipsoidal joint confidence region for the SAFE AI metrics, as follows:
\begin{equation}
\left\{\theta :
(\thetahat-\theta)^{\top}\Sighat^{-1}(\thetahat-\theta)
\leq \chi^2_{3,\,1-\alpha}\right\},
\label{eq:region}
\end{equation}
where $\chi^2_{3,\,1-\alpha}$ is the $(1-\alpha)$ quantile of the chi-squared distribution with three degrees of freedom. This is an asymptotic joint confidence region based on the estimated bootstrap covariance.

We have verified the calibration of both constructions in a known-truth simulation with nine scenario configurations, varying the test-sample size from 100 to 800 rows, the sign and strength of the latent correlations, and the perturbation level, with 1000 outer replications per scenario; the Monte Carlo standard error at the nominal 95\% level is about 0.7 percentage points. The chi-squared region of \eqref{eq:region} attains empirical coverage between 92.7\% and 95.2\%, and the paired percentile interval for the composite score between 94.1\% and 97.1\%. The same percentile interval built from independent bootstrap streams, which mimics the independence baseline, ranges between 80.2\% and 99.8\% across the same scenarios: anticonservative where the correlations are positive and wastefully wide where they are negative. This motivates the paired construction used throughout the paper.

\subsection{Compliance confidence intervals }\label{sec:theory}

To obtain a confidence interval for  the compliance score $C$, we can apply  the delta method, and obtain $\operatorname{var}(\hat C) \approx \nabla g^{\top} \Sighat\, \nabla g$, where $\nabla g$ is the gradient of $g$ evaluated at $\thetahat$, and  $\Sighat$ is approximated  via the paired-bootstrap \citep{efrontibshirani1993}.

More precisely, let $\thetahat\in\mathbb{R}^p$ be a vector of  metrics estimated on a common test sample, with positive-definite estimated covariance matrix $\Sighat$. Let $\sigma_k=\Sighat_{kk}^{1/2}>0$ denote the marginal standard error of metric $k$, let $D=\operatorname{diag}(\sigma_1,\ldots,\sigma_p)$, and let $P=D^{-1}\Sighat D^{-1}$ be the resulting correlation matrix. Let $C=g(\theta)$ be a compliance score and $\hat C=g(\thetahat)$ its plug-in estimate, where $g$ is differentiable at $\thetahat$, and set $\nabla=\nabla g(\thetahat)\neq0$.

The delta method gives the approximate variance
\begin{equation}
v=\nabla^{\top}\Sighat\nabla,
\qquad
v_0=\sum_{k=1}^{p}\nabla_k^2\Sighat_{kk},
\label{eq:twovars}
\end{equation}
where $v_0$ is obtained under the assumption of independence, that is, setting all cross-covariances to zero. 

Let $R=\sqrt{v/v_0}>0$ be the ratio of the full-covariance interval width to the diagonal-only interval width.

\begin{definition}[Exposure]
\label{def:exposure}
The exposure vector is $d=D\nabla$, with components
$d_k=\nabla_k\sigma_k$. Its normalised form is
$u=d/\lVert d\rVert_2$, so that
$u_k=\nabla_k\sigma_k/\lVert d\rVert_2$ and
$\lVert u\rVert_2=1$. Define the effective number of exposed components by
\begin{equation}
p_{\mathrm{eff}}
=
\left(\sum_{k=1}^p u_k\right)^2,
\label{eq:peff}
\end{equation}
and, whenever $\sum_{k<l}u_ku_l\neq0$, define the
\emph{exposure-weighted average correlation} by
\begin{equation}
\bar{\rho}
=
\frac{\displaystyle\sum_{k<l}u_ku_l\,\rho_{kl}}
{\displaystyle\sum_{k<l}u_ku_l},
\label{eq:rhoeff}
\end{equation}
where $\rho_{kl}$ is the $(k,l)$-th entry of $P$.
\end{definition}

Since $v_0=\sum_kd_k^2=\lVert d\rVert_2^2$, the quantity
$u_k^2=d_k^2/v_0$ is the share of the diagonal-only variance contributed by
component $k$. If $\nabla\geq0$, then $u\geq0$ and
$1\leq p_{\mathrm{eff}}\leq p$: the lower bound is attained when only one
component has non-zero exposure, and the upper bound when all exposures are
equal.

\begin{proposition}[Width ratio decomposition]\label{prop:main}
For every non-zero exposure vector,
\begin{equation}
R^2
=\frac{v}{v_0}
=u^{\top}Pu
=1+2\sum_{k<l}u_ku_l\rho_{kl}.
\label{eq:main}
\end{equation}
Whenever $\bar\rho$ is defined, this can equivalently be written as
\begin{equation}
R^2=1+\bigl(p_{\mathrm{eff}}-1\bigr)\bar\rho.
\label{eq:main-rho}
\end{equation}
\end{proposition}

\begin{proof}
By the definitions above,
$v=\nabla^{\top}\Sighat\nabla=d^{\top}Pd
=\lVert d\rVert_2^2u^{\top}Pu$, while
$v_0=\sum_kd_k^2=\lVert d\rVert_2^2$. Dividing and expanding the quadratic
form gives Equation~\eqref{eq:main}. Finally,
$2\sum_{k<l}u_ku_l=(\sum_k u_k)^2-\sum_k u_k^2
=p_{\mathrm{eff}}-1$, which gives Equation~\eqref{eq:main-rho} whenever
$\bar\rho$ is defined.
\end{proof}

\begin{definition}[Relative width and variance width effects]\label{cor:U}
Let
$\omega=2z_{1-\alpha/2}\sqrt v$ and
$\omega_0=2z_{1-\alpha/2}\sqrt{v_0}$
be the widths of the full-covariance and diagonal-only normal intervals,
respectively, where $z_{1-\alpha/2}$ is the corresponding standard normal
quantile. Define their signed relative width and variance effects by
\begin{equation}
U=1-\frac{\omega_0}{\omega}=1-R^{-1},
\qquad
V=1-\frac{v_0}{v}=1-R^{-2}.
\label{eq:UV}
\end{equation}
Whenever $\bar\rho$ is defined, Proposition~\ref{prop:main} gives
\begin{equation}
U=1-\bigl[1+(p_{\mathrm{eff}}-1)\bar\rho\bigr]^{-1/2},
\qquad
V=\frac{(p_{\mathrm{eff}}-1)\bar\rho}
{1+(p_{\mathrm{eff}}-1)\bar\rho}.
\label{eq:U}
\end{equation}
In every case, $U=1-\sqrt{1-V}$, or equivalently,
$V=1-(1-U)^2$.
\end{definition}

Positive values of $U$ and $V$ mean that the diagonal-only calculation
understates uncertainty; negative values mean that it overstates uncertainty.
Thus, $U$ and $V$ express the same effect on the interval-width and variance
scales, respectively.

We now specialize the width ratio decomposition in Proposition 1 to four common metric aggregators.

\begin{corollary}[Linear aggregator]\label{cor:linear}
If $g(\theta) = w^{\top}\theta$ for a fixed weight vector $w$, which includes the arithmetic mean, then $\nabla = w$ is constant. For a given true covariance matrix $\Sigma$, $\operatorname{Var}(\hat C)=w^{\top}\Sigma w$, while $v=w^{\top}\Sighat w$ is the corresponding plug-in estimate.
Moreover, $u_k = w_k\sigma_k / \lVert w \circ \sigma\rVert_2$, where $\circ$ denotes element-wise multiplication. If, in addition, the
exposures are balanced, $w_1\sigma_1 = \dots = w_p\sigma_p$, then
$p_{\mathrm{eff}} = p$ and
\begin{equation}
R^2 \;=\; 1 + (p-1)\,\bar\rho ,
\label{eq:sb}
\end{equation}
with $\bar\rho$ the unweighted mean of the $\binom{p}{2}$ correlations.
\end{corollary}

Under balanced exposures, Equation~\eqref{eq:sb} recovers the familiar correlation-adjustment factor $1+(p-1)\bar\rho$, which also appears in the denominator of the Spearman--Brown reliability formula (see \cite{Spearman_Brown}). 

\begin{corollary}[Multiplicative aggregator]\label{cor:mult}
If $g(\theta) = \prod_k \theta_k^{b_k}$, with $\theta_k>0$ and $b_k\geq0$, which includes the geometric mean,
then $\nabla_k=b_k\hat C/\thetahat_k$ and $d_k=\hat C b_k\mathrm{CV}_k$, where $\mathrm{CV}_k=\sigma_k/\thetahat_k$ is the relative standard error (coefficient of variation) of component $k$. For equal positive exponents $u$ is the normalised coefficient-of-variation vector and 
\begin{equation}
p_{\mathrm{eff}}
=
\frac{(\sum_k\mathrm{CV}_k)^2}{\sum_k\mathrm{CV}_k^2}.
\end{equation}
\end{corollary}

Note that for equal exponents Corollary~\ref{cor:mult} makes the diagnostic particularly simple for the three metrics we consider here: their three coefficients of variation determine $p_{\mathrm{eff}}$, and $p_{\mathrm{eff}}$, together with $\bar\rho$, gives $U$.

\begin{corollary}[Min-type aggregator]\label{cor:onehot}
If $\nabla = c\,\mathbf{e}_j$ for some $j$ and a constant $c \neq 0$, with $\mathbf{e}_j$ the $j$-th standard basis vector, then $d = c\sigma_j \mathbf{e}_j$,
$u = \pm \mathbf{e}_j$, $p_{\mathrm{eff}} = 1$, and
\begin{equation}
R = 1, \qquad U = V = 0
\label{eq:onehot}
\end{equation}
\end{corollary}

This includes composites that are locally determined by a unique order
statistic, such as $\min_k\theta_k$ away from ties. The result holds exactly for the two first-order variance expressions, and for every $\Sighat$, regardless of the dependence among the components.

\begin{corollary}[Worst-case  aggregator]\label{cor:bound}
Suppose that $p\geq2$ and $\nabla\geq0$, and define
$\rho^{+}=\max\{0,\max_{k<l}\rho_{kl}\}$.
Then
\begin{equation}
R\leq\sqrt{1+(p-1)\rho^{+}},
\qquad
U\leq1-\bigl(1+(p-1)\rho^{+}\bigr)^{-1/2}.
\label{eq:bound}
\end{equation}
Equality holds if every pairwise correlation equals $\rho^{+}$ and the exposures are balanced. When $\rho^{+}>0$, these conditions are also necessary. If, in addition, every entry of $P$ is non-negative, the supremum over all non-negative gradients is attained and is given by the leading eigenvalue:
\begin{equation}
\sup_{\nabla\geq0}R^2=\lambda_{\max}(P),
\qquad
\sup_{\nabla\geq0}U
=1-\lambda_{\max}(P)^{-1/2}.
\label{eq:lam}
\end{equation}
A maximising gradient is
$\nabla^{\star}=D^{-1}v_{\max}$,
where $v_{\max}$ is the non-negative unit eigenvector of $P$ associated with $\lambda_{\max}(P)$.
\end{corollary}

\begin{proof}
$\nabla \ge 0$ and $\sigma > 0$ give $u \ge 0$, so $u_k u_l \ge 0$ and
$u^{\top}Pu = 1 + 2\sum_{k<l}u_k u_l\rho_{kl}
\le 1 + \rho^{+}(p_{\mathrm{eff}}-1) \le 1 + (p-1)\rho^{+}$,
the last step by Cauchy--Schwarz, $(\sum_k u_k)^2 \le p\lVert u\rVert^2 = p$.
For~\eqref{eq:lam}: as the non-zero vector $\nabla$ ranges over the non-negative orthant, $u$ ranges over the non-negative part of the unit sphere. For any unit vector $x$, entrywise non-negativity of $P$ gives $x^{\top}Px\leq\lvert x\rvert^{\top}P\lvert x\rvert$. Hence, the
constrained and unconstrained maxima coincide at $\lambda_{\max}(P)$,
which is attained by a non-negative eigenvector according to the
Perron--Frobenius theorem. %
\end{proof}

Note that inverting~\eqref{eq:bound} gives a screening rule. To ensure that ignoring
the cross-terms understates the interval width by at most a fraction $\varepsilon\in(0,1)$, it is sufficient that
\begin{equation}
\max_{k<l}\rho_{kl} \;\le\; \frac{(1-\varepsilon)^{-2}-1}{p-1}.
\label{eq:screen}
\end{equation}
For $p=3$ and $\varepsilon = 0.05$ this leads to $\rho_{\max} \le 0.054$; for
$\varepsilon = 0.10$, $\rho_{\max}\le 0.117$. 

We conclude this section demonstrating what is the cost and the sign of the independence approximation.

\begin{proposition}[Cost of independence]
\label{cor:sign}
Let
\begin{equation}
\Delta_{\mathrm{cov}}
:=v-v_0
=2\sum_{k<l}\nabla_k\nabla_l\widehat{\Sigma}_{kl}.
\label{eq:delta-cov}
\end{equation}
If $\nabla\geq0$ and
$p_{\mathrm{eff}}>1$, then
\begin{equation}
\Delta_{\mathrm{cov}}
=v_0\bigl(p_{\mathrm{eff}}-1\bigr)\bar\rho,
\qquad
\operatorname{sign}(U)=\operatorname{sign}(\bar\rho).
\label{eq:sign-rho}
\end{equation}
This implies that, for a monotone compliance score exposed to at least two components, the cost of independence is: a narrower confidence interval,  when
the exposure-weighted correlation is positive;  a wider confidence interval, when it is negative. 
Note that, when $p_{\mathrm{eff}}=1$, exactly one component has
non-zero exposure, so $R=1$ and $U=0$ irrespective of the correlations. 
\end{proposition}

\begin{proof}
Equation~\eqref{eq:delta-cov} follows by expanding
$v=\nabla^{\top}\Sighat\nabla$. The width comparison follows because the
two widths are proportional to $\sqrt v$ and $\sqrt{v_0}$, and
$U=1-\sqrt{v_0/v}$. If $\nabla\geq0$, then $u\geq0$ and
$2\sum_{k<l}u_ku_l=p_{\mathrm{eff}}-1$. When this sum is zero, $u$ is
one-hot encoded and $R^2=u^{\top}Pu=1$. When it is positive, the definition of
$\bar\rho$ gives
$\Delta_{\mathrm{cov}}=v_0(p_{\mathrm{eff}}-1)\bar\rho$, proving the
remaining claims.
\end{proof}

\subsection{A reproducible uncertainty certificate}\label{sec:cert}

A reported measurement score should be accompanied by the
conditions under which it was measured, in the same way as a  calibrated instrument should
travel with its calibration certificate. 

A certified compliance score should describe the conditions under which it was measured.
Measurement institutions address this problem with
calibration certificates which state the measurement method, the measured value, the found uncertainty, 
the validity horizon --- and the mapping patterns which transfer the result to a
compliance score essentially unchanged
\citep{jcgm100, dcc}. The certificate should be signed, dated and
machine-readable. Its subject fields should identify the model, the substrate version,
the dataset and test-sample specifications. Its measurement fields should carry the estimates of the component metrics,
the estimated covariance among them, the composed compliance score, the corresponding confidence interval, along the  method
that produced it, including the bootstrap settings, seed and perturbation family. 
A recalibrate-by date and a
signature over the canonical body make the record time-bounded and
tamper-evident. The signing-and-binding pattern can follow prior work on evidence
packages for AI-agent actions \citep{sokolov2026aep}.

As a worked example, the certificate issued for the executed two-link chain of Section~\ref{sec:simulated} records a composed score of $0.808$ with a $95\%$ interval of $[0.768, 0.847]$, propagated to first order from the paired-bootstrap link covariance ($B = 2000$; the method, seeds, perturbation family and substrate commit travel in the certificate body), with a recalibrate-by date ninety days after issue. Screened against an illustrative policy threshold of $0.75$ the whole interval clears and the action is allowed, while a threshold falling inside the interval, such as $0.80$, escalates the decision rather than allowing or denying it (Figure~\ref{fig:cert-example}). The certificate also records the zero-cross-term reading $[0.773, 0.843]$, which is $11.0\%$ narrower, so a consumer can see exactly which policy thresholds the independence shortcut would decide differently.

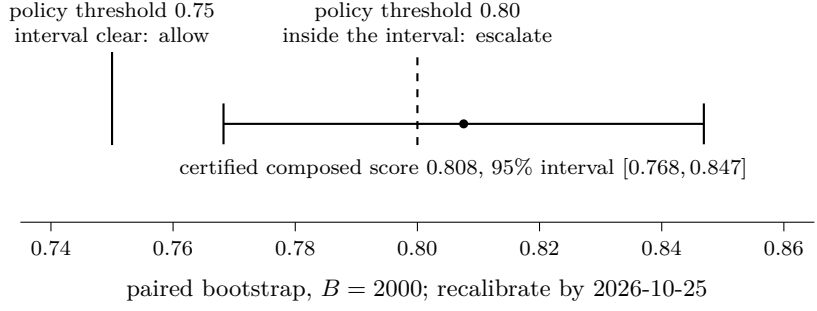
\begin{figure}[t]\centering
\begin{tikzpicture}[y=10mm]

  \draw[thick] (12.12mm,-0.28) -- (12.12mm,0.95);
  \node[anchor=south,font=\scriptsize,align=center] at (12.12mm,0.95)
    {policy threshold $0.75$\\interval clear: allow};

  \draw[thick,dashed] (52.50mm,-0.28) -- (52.50mm,0.95);
  \node[anchor=south,font=\scriptsize,align=center] at (52.50mm,0.95)
    {policy threshold $0.80$\\inside the interval: escalate};

  \draw[thick] (26.86mm,0.00) -- (90.38mm,0.00);
  \draw[thick] (26.86mm,-0.27) -- (26.86mm,0.27);
  \draw[thick] (90.38mm,-0.27) -- (90.38mm,0.27);
  \fill (58.62mm,0.00) circle (1.7pt);
  \node[anchor=north,font=\scriptsize] at (58.62mm,-0.32)
    {certified composed score $0.808$, 95\% interval $[0.768, 0.847]$};

  \draw (0.00mm,-1.30) -- (105.00mm,-1.30);
  \draw (4.04mm,-1.30) -- (4.04mm,-1.42);
  \node[anchor=north,font=\scriptsize] at (4.04mm,-1.42) {$0.74$};
  \draw (20.19mm,-1.30) -- (20.19mm,-1.42);
  \node[anchor=north,font=\scriptsize] at (20.19mm,-1.42) {$0.76$};
  \draw (36.35mm,-1.30) -- (36.35mm,-1.42);
  \node[anchor=north,font=\scriptsize] at (36.35mm,-1.42) {$0.78$};
  \draw (52.50mm,-1.30) -- (52.50mm,-1.42);
  \node[anchor=north,font=\scriptsize] at (52.50mm,-1.42) {$0.80$};
  \draw (68.65mm,-1.30) -- (68.65mm,-1.42);
  \node[anchor=north,font=\scriptsize] at (68.65mm,-1.42) {$0.82$};
  \draw (84.81mm,-1.30) -- (84.81mm,-1.42);
  \node[anchor=north,font=\scriptsize] at (84.81mm,-1.42) {$0.84$};
  \draw (100.96mm,-1.30) -- (100.96mm,-1.42);
  \node[anchor=north,font=\scriptsize] at (100.96mm,-1.42) {$0.86$};
  \node[anchor=north,font=\footnotesize] at (52.50mm,-1.92)
    {paired bootstrap, $B = 2000$; recalibrate by 2026-10-25};

\end{tikzpicture}
\caption{A worked example of the uncertainty certificate: the certified composed score of the executed two-link chain of Section~\ref{sec:simulated}, its $95\%$ interval, and two illustrative policy thresholds --- one the whole interval clears (allow) and one falling inside the interval (escalate). All quantities except the two thresholds are read from the certificate.}
\label{fig:cert-example}
\end{figure}

\section{Application to credit lending}
\label{sec:german}
\label{sec:chain}

We now consider a real data example: the well-known German credit lending data.

Four compliance-score aggregation rules are applied to paired bootstrap replicates of the three SAFE AI metrics: the arithmetic mean, geometric mean, root-mean-square mean, and the Technique for Order Preference by Similarity to Ideal Solution (TOPSIS) of Hwang and Yoon \cite{hwangyoon1981}. The TOPSIS variant uses prespecified ideal and anti-ideal vectors, equal weights, and Euclidean distances without data-dependent column normalisation.

For the German credit data, the explainability curve fixes the grid
resolution because the data contain $d=20$ features. We use the common grid $s_t=t/d$, $t=0,\ldots,d$, giving $L=d+1=21$ knots. For the robustness curve, this grid corresponds to tail-swap severity
$p_t=0.5t/d$. We retain both endpoints, so the grid contains one more knot than stated in the defining article. This is an explicit implementation choice of the present analysis.

Our primary analysis follows \cite{giudici2026integrated} and uses
tail-swap perturbations. We also evaluate a scaled-Gaussian sweep to examine the sensitivity of the results to the perturbation family.

Table~\ref{tab:volume} reports the Compliance Score estimates and their paired-bootstrap intervals. The interval widths range from $0.044$ to $0.075$. The corresponding delta-method widths differ from the bootstrap widths by at most $3.1\%$ across all eight model--aggregator combinations.

\begin{table}[h]
\centering
\caption{Compliance Score estimates and paired-bootstrap intervals under tail-swap robustness. The quantities $100U$ and $100V$ give the signed percentage effects of setting the cross-covariances to zero on the delta-method interval width and variance, respectively. Positive values indicate understated uncertainty and Monte Carlo standard errors are shown in parentheses.}
\label{tab:volume}
\begin{tabular}{lcrrrr}
\toprule
Aggregator & $\widehat{\mathcal V}$ & Bootstrap 95\% & Width
& $100U$ (SE) & $100V$ (SE) \\
\midrule
\multicolumn{6}{l}{\textit{Logistic regression}} \\
Arithmetic & 0.631 & $[0.603,0.661]$ & 0.058 & $-0.91$ (0.73) & $-1.84$ (1.48) \\
Geometric & 0.535 & $[0.492,0.567]$ & 0.075 & $-1.86$ (0.98) & $-3.76$ (2.00) \\
RMS & 0.699 & $[0.678,0.725]$ & 0.046 & $+1.26$ (0.54) & $+2.51$ (1.07) \\
TOPSIS & 0.380 & $[0.357,0.406]$ & 0.048 & $-1.50$ (0.63) & $-3.02$ (1.29) \\
\midrule
\multicolumn{6}{l}{\textit{Random forest}} \\
Arithmetic & 0.672 & $[0.646,0.700]$ & 0.054 & $-3.08$ (0.89) & $-6.25$ (1.84) \\
Geometric & 0.601 & $[0.567,0.632]$ & 0.066 & $-3.09$ (1.02) & $-6.27$ (2.10) \\
RMS & 0.721 & $[0.701,0.746]$ & 0.045 & $-1.50$ (0.69) & $-3.02$ (1.39) \\
TOPSIS & 0.415 & $[0.394,0.438]$ & 0.044 & $-3.11$ (0.84) & $-6.32$ (1.73) \\
\bottomrule
\end{tabular}
\end{table}

Assuming independence changes the interval width by between $-3.11\%$ and $+1.26\%$, with corresponding variance differences between $-6.32\%$ and $+2.51\%$. The approximation is mostly conservative because it increases the interval width, but not uniformly so. It produces a narrower interval only for the logistic regression RMS aggregator. In general, the effect of ignoring dependence varies with the machine learning model and the aggregation rule.
The Root Mean Square (RMS) mean results for the logistic regression model in Table~\ref{tab:volume} is the  configuration in which assuming independence produces a narrower interval. This exception is consistent with the RMS gradient, which gives more weight to knots with larger values. This  concentrates exposure at lower severity levels where the relevant correlations remain positive.

It is of interest to measure the correlations between the metrics. To this end, Table~\ref{tab:bridge} presents the obtained correlation between RGR and RGE,  while holding the bootstrap resamples fixed.

\begin{table}[h]
\centering
\caption{$\rge$--$\rgr$ correlations for alternative estimator pairs, based on the same fixed-draw bootstrap resamples.}
\label{tab:bridge}
\small
\begin{tabular}{llcc}
\toprule
\multicolumn{2}{c}{$\rge \times \rgr$ estimator pair}
& \multicolumn{2}{c}{Correlation} \\
\cmidrule(lr){1-2}\cmidrule(lr){3-4}
$\rge$ estimator & $\rgr$ estimator
& Logistic regression & Random forest \\
\midrule
Single-mask mean & Fixed Gaussian draw
& $+0.597$ & $+0.342$ \\
Single-mask mean (probability scoring) & Fixed Gaussian draw
& $+0.614$ & $+0.344$ \\
Greedy-curve mean & Fixed Gaussian draw
& $+0.425$ & $+0.288$ \\
Single-mask mean & Tail-swap curve mean
& $-0.208$ & $-0.153$ \\
Greedy-curve mean & Gaussian-sweep mean
& $+0.535$ & $+0.471$ \\
Greedy-curve mean & Tail-swap curve mean
& $-0.137$ & $-0.310$ \\
\bottomrule
\end{tabular}
\end{table}

Table~\ref{tab:bridge}  shows that 
replacing the single-mask explainability estimator with the greedy-curve mean reduces the positive correlation but does not reverse its sign under Gaussian perturbations. By contrast, combining the explainability estimator with the tail-swap curve produces a negative correlation.  Replacing tail swaps with the scaled-Gaussian sweep also increases the correlation. Thus, the perturbation family affects both the level of the compliance score and the  covariance correction.

For robustness purposes, we varied the calculation method of the compliance score.
Equation~\eqref{eq:volume} uses full-product indexing: it averages
$M(i,j,k)$ over all combinations of the three indices, so the accuracy, explainability, and robustness values in a given term need not correspond to the same severity. An alternative is matched-severity indexing, which averages only the diagonal terms:
\[
\mathcal V_{\mathrm{diag}}
=
\frac{1}{L}\sum_{t=1}^{L}m(a_t,e_t,r_t).
\]
We have computed both scores and estimate their difference with the paired bootstrap.
The results show that, for the arithmetic mean, the two indexing rules are algebraically equivalent, and their values agree within $10^{-16}$ in every replicate. For the geometric mean, the matched-severity score is higher by $0.022$--$0.025$ across the two models.  For the root mean square mean, the matched-severity score is lower by $0.008$--$0.009$, corresponding to approximately one fifth of the interval width.
This pattern is consistent with the co-monotonicity of the aligned curves. Matched-severity indexing combines high values with high values and low values with low values. In the present data, this increases the geometric aggregate and decreases the root-mean-square aggregate. We remark that the matched-severity interpretation is available only when the three curves share a common grid; if separate grids are used, equal indices do not represent equal severity.

We now consider a worst-case aggregator. All entries of the scalar-metric correlation matrices in
Table~\ref{tab:corr} are non-negative. Corollary~\ref{cor:bound} therefore gives the worst case over non-negative weightings in closed form. The maximum discarded variance share is $1-\lambda_{\max}(P)^{-1}$, and the maximum width understatement is $1-\lambda_{\max}(P)^{-1/2}$.
For the logistic regression model, $\lambda_{\max}(P)=1.5803$, giving a maximum discarded variance share of $36.72\%$ and a maximum width
understatement of $20.45\%$. For the random forest model,
$\lambda_{\max}(P)=1.5143$, giving corresponding values of $33.96\%$ and $18.74\%$. With coordinates ordered as $(\rga,\rge,\rgr)$, the maximizing weightings are proportional to $(0.019,0.873,0.107)$ and
$(0.025,0.918,0.057)$, respectively. 
The worst-case value is determined by the correlation matrix $P$, while the corresponding maximizing weightings are affected by the unequal marginal variability of the three metrics. These worst-case figures are computed under the Gaussian perturbation family, for which all pairwise correlations remain non-negative on our data; under the tail-swap family some correlations turn negative and the non-negativity premise behind the closed-form bound no longer applies. Under redrawing, the bootstrap standard errors of $\rga$, $\rge$, and $\rgr$ are $0.0294$, $0.00095$, and $0.0069$, respectively, for the logistic regression model, with the same general ordering for the random forest model. As a result, the $\rga$--$\rgr$ pair makes the largest cross-term contribution even though the $\rge$--$\rgr$ correlation is larger. For a fixed compliance score,  variance contributions depend on the exposures and covariances, not on the correlations alone.

 Consider next a min-type aggregator. As it has a unique active component, Corollary~\ref{cor:onehot} gives a one-hot gradient and hence $U=V=0$. This does not mean that the component metrics are independent. It means that the other components do not affect the score locally, so their covariances do not enter its first-order variance.
The situation changes at a tie. If two or more components attain the
minimum, the functional is directionally differentiable but not ordinarily differentiable, and there is no unique gradient to use in the delta method. Moreover, the standard nonparametric bootstrap may not be consistent at such points \citep{duembgen1993,fangsantos2019}.
For example, if we increase the perturbation scale until the robustness and accuracy metrics cross, the bootstrap interval becomes narrower than the delta interval by $27\%$ for the logistic regression model and $23\%$ for the random forest model. Because neither procedure is guaranteed to be valid at an exact tie, these differences demonstrate instability rather than calibrated uncertainty.
Thus, for a min-type composite near a tie, an operator should inspect the component intervals or use an inference method designed for directionally differentiable functionals rather than relying on the ordinary delta or bootstrap interval.

In summary, the real data analysis shows that the effect of ignoring dependence is determined jointly by the aggregation rule, the uncertainty-weighted exposures, and the correlation structure. The perturbation and aggregating rule choices can also affect the score, while min-type composites require special care near ties.

\section{Application to agentic AI}
\label{sec:simulated}
We now consider the application of our methodology to simulated data, for which we assume an agentic research pipeline.
In the context of agentic AI, each agent executes a task, and could then be evaluated independently from the others. However, as the final result depends on the composition of the tasks of different agents, the evaluation should also be composite.

For the sake of illustration, we first consider a stylised system composed of two agents. Agent A makes its decisions using a logistic regression model, whereas Agent B employs a random forest model. The overall agentic chain compliance score is therefore $h = C_A \cdot C_B$. For simplicity, each link score is the geometric mean of the three scalar metrics, and both links are evaluated on the same 300 test rows of the German credit data of the previous section. This construction provides a statistical illustration of dependence between link scores sharing an evaluation sample rather than an executed agentic interaction.

With perturbations redrawn in  2{,}000 paired-bootstrap
replicates, the link scores are $C_A = 0.913$, with 95\% interval $[0.889, 0.935]$, and $C_B = 0.906$, with interval $[0.884, 0.930]$. Their estimated correlation is $\rho(C_A,C_B)=0.712$, and the composite compliance score is:
\[
h=0.913\times0.906=0.827.
\]
Including the estimated covariance leads to a  95\% interval equal to 
$[0.787,0.867]$, with width $0.0800$. Assuming independence, that is, setting the cross-link covariance to zero, gives $[0.796,0.857]$, with width $0.0611$, understating the interval width by $23.6\%$. Under a fixed-draw scheme, the correlation is $0.722$ and the width understatement is $23.8\%$. The cost of independence is therefore stable across the two conditioning schemes.

Figure \ref{fig:chain} pictorially represents the obtained results.

\begin{figure}[H]
\centering
\begin{tikzpicture}[y=10mm]

  \fill[pattern=north east lines, pattern color=black!45]
    (11.09mm,-1.70) rectangle (20.08mm,1.45);
  \draw[dotted] (11.09mm,-1.70) -- (11.09mm,1.45);
  \draw[dotted] (20.08mm,-1.70) -- (20.08mm,1.45);

  \draw[thick] (15.27mm,-1.80) -- (15.27mm,1.62);
  \node[anchor=south,font=\footnotesize] at (15.27mm,1.62)
    {operator threshold $0.791$ (illustrative)};

  \draw[thick] (11.09mm,0.85) -- (87.41mm,0.85);
  \draw[thick] (11.09mm,0.58) -- (11.09mm,1.12);
  \draw[thick] (87.41mm,0.58) -- (87.41mm,1.12);
  \fill (49.25mm,0.85) circle (1.7pt);
  \node[anchor=east,font=\footnotesize,align=right] at (7.09mm,0.85)
    {(a) measured\\cross-link\\covariance};
  \node[anchor=south,font=\scriptsize] at (49.25mm,1.05)
    {point estimate $0.827$, width $0.0800$};
  \node[anchor=north,font=\scriptsize] at (49.25mm,0.60)
    {straddles the threshold: undecided};

  \draw[thick,dashed] (20.08mm,-1.05) -- (78.41mm,-1.05);
  \draw[thick] (20.08mm,-1.25) -- (20.08mm,-0.85);
  \draw[thick] (78.41mm,-1.25) -- (78.41mm,-0.85);
  \draw (49.25mm,-1.05) circle (1.7pt);
  \node[anchor=east,font=\footnotesize,align=right] at (7.09mm,-1.05)
    {(b) cross-term\\declared zero};
  \node[anchor=south,font=\scriptsize] at (54.25mm,-0.87)
    {width $0.0611$, understating it by $23.6$ per cent};
  \node[anchor=north,font=\scriptsize] at (49.25mm,-1.27)
    {clear of the threshold: reads as passing};

  \node[anchor=north,font=\scriptsize] at (49.25mm,-1.78)
    {any threshold in the hatched band: (a) is undecided, (b) passes};

  \draw (0.00mm,-2.30) -- (105.00mm,-2.30);
  \draw (4.77mm,-2.30) -- (4.77mm,-2.42);
  \node[anchor=north,font=\scriptsize] at (4.77mm,-2.42) {$0.78$};
  \draw (23.86mm,-2.30) -- (23.86mm,-2.42);
  \node[anchor=north,font=\scriptsize] at (23.86mm,-2.42) {$0.80$};
  \draw (42.95mm,-2.30) -- (42.95mm,-2.42);
  \node[anchor=north,font=\scriptsize] at (42.95mm,-2.42) {$0.82$};
  \draw (62.05mm,-2.30) -- (62.05mm,-2.42);
  \node[anchor=north,font=\scriptsize] at (62.05mm,-2.42) {$0.84$};
  \draw (81.14mm,-2.30) -- (81.14mm,-2.42);
  \node[anchor=north,font=\scriptsize] at (81.14mm,-2.42) {$0.86$};
  \draw (100.23mm,-2.30) -- (100.23mm,-2.42);
  \node[anchor=north,font=\scriptsize] at (100.23mm,-2.42) {$0.88$};
  \node[anchor=north,font=\footnotesize] at (52.50mm,-2.92)
    {composed compliance figure $h = C(\mathrm{logit}) \times C(\mathrm{rf})$};

\end{tikzpicture}
\caption{Compliance score 95\% intervals for the agentic score $h=C_A C_B$ when the measured cross-link covariance is included and when it is set to zero. Both intervals are built around the estimate $0.827$. The hatched region marks thresholds for which the covariance-aware interval is undecided while the zero-covariance interval gives a pass.}
\label{fig:chain}
\end{figure}
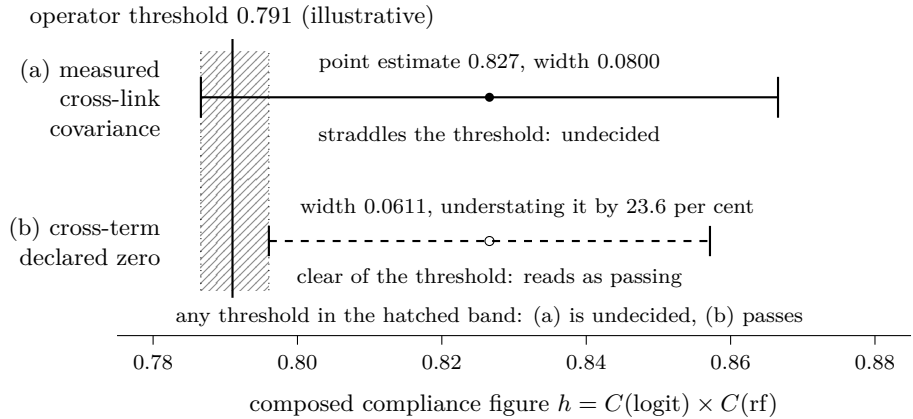

The decision of whether assigning a uncertainty certificate to the operator running an agentic chain as the above depends on the  decision threshold.
The decision threshold reflects the operator's risk appetite. Together with the confidence interval for the composite score, it produces a three-valued verdict: pass if the interval lies entirely above the threshold, fail if it lies entirely below, and undecided if it crosses the threshold.

For a threshold in the lower hatched region of
Figure~\ref{fig:chain}, the interval obtained under independence gives a pass, whereas the interval carrying the measured covariance remains undecided. Although not shown in Figure~\ref{fig:chain}, a corresponding disagreement occurs between the upper endpoints: the independence interval gives a fail while the covariance-aware interval remains undecided. Most thresholds lie outside these narrow regions and therefore produce the same verdict. Within them, however, marginal link intervals give a definitive conclusion that is not supported after the cross-link covariance is included.

This indicates that evaluating agents independently is insufficient to establish compliance; the  covariance between agentic evaluations must also be reported and included.

We now consider a different agentic setting. We assume that an agent acting on someone's behalf is checked twice, at two different
moments, by two parties that do not share state. Before the action, an
authorisation gate checks whether the agent may take an action. After
the action, another verifier checks whether the action matches what
was authorised. The two checks read overlapping evidence -- the same
request, the same mandate -- but they run at different times, and neither
sees the other's verdict. The open question is how to combine the two checks in a single
compliance statement and, in particular, whether they can be assumed independent.

To answer this question, we simulated a pipeline of 1,000 randomised episodes, including 508 clean episodes and 492 episodes carrying faults drawn from a sixteen-channel taxonomy based on the pipeline's published conformance vectors and adversarial test suites. Stage A is the pre-action authorisation gate, while Stage B is the post-action verifier. At each stage, the primary model is a standard scaler followed by logistic regression, and a 300-tree random forest is used as a sensitivity model. The models are trained on 700 episodes and evaluated on the remaining 300 using seed 20260727.

Stage A uses nine pre-decision explanatory variables covering the requested tool, amount, payee class, mandate scope, spend-cap headroom, expiry margin, and issuer status. Stage B uses nine variables covering the same request metadata, together with the evidence-package size and two nonce matches. Signature-validity and recomputed-hash indicators remain with the verifiers and are not included in the models. To calculate the RGR metric, Gaussian feature noise with featurewise standard deviation equal to $0.5$ times the corresponding evaluation-sample column standard deviation is redrawn within each of the 2,000 paired row-bootstrap replicates.

The simulator assigns about half of the response actions to the clean channel and the rest across the 16 fault channels.
To populate all four stage-A/stage-B verdicts, we let an action
denied at stage A, after producing a sealed evidence record, proceed to stage B in any case.

The corresponding logistic regression compliance scores are $C_A=0.937$, with a 95\% interval of $[0.915,0.954]$, and $C_B=0.862$, with an interval of $[0.828,0.897]$, under redrawn perturbations. The composite score is $h=0.808$, and the estimated cross-link correlation is $+0.314$. The resulting correlation-aware
95\% interval is $[0.768,0.847]$. Assuming independence gives an interval of $[0.773,0.843]$, understating the interval width by $11.0\%$.

Across the fixed-draw and redraw schemes and both machine learning models, the compliance value ranges from $0.805$ to $0.808$, the correlation from $+0.24$ to $+0.32$, and the width understatement from $8.5\%$ to $11.0\%$. This dependence is weaker than in the German-credit illustration, where the correlation is $+0.712$ and the width understatement is $23.6\%$. Nevertheless, it remains material, accounting for approximately one tenth of the interval width, and it remains invisible to any compliance certificate based only on marginal intervals.

We remark that the composition level is not where the cost of independence is largest. Within Stage A, assuming independence understates the stage interval by $15.0 \pm 0.7\%$ for logistic regression and $11.8 \pm 0.8\%$ for random forest. Within Stage B, it understates the stage interval by $12.7 \pm 0.5\%$ for logistic regression and $19.9 \pm 0.5\%$ for the random forest. Thus, the within-link costs exceed the cross-link costs in all four cases.

We now repeat the analysis at each stage using the curve-based volume
composites, instead of the scalar-metric composites considered above.

For each stage and machine learning model, we reconstruct the three curves on a common severity grid. Each stage has $d=9$ features, so we use $s_t=t/9$, $t=0,\ldots,9$, giving $L=10$ knots. These correspond to accuracy after removing $t$ top-ranked segments, explainability after greedily removing $t$ features, and robustness at tail-swap severity $p_t = 0.5\,t/9$. The greedy removal order, tail-swap quantiles, and masking column means are held fixed at their full-sample values. Consequently, the reported intervals are conditional on these choices. Table~\ref{tab:volume-exec} reports the four volume composites for each stage and model.

\begin{table}[h]
\centering
\caption{Compliance Score estimates and paired-bootstrap intervals for the simulated agentic stages under tail-swap robustness. The quantities $100U$ and $100V$ are the signed percentage effects of setting the cross-covariances to zero on the delta-method interval width and variance, respectively. Positive values indicate understated uncertainty. Monte Carlo standard errors are shown in parentheses.}
\label{tab:volume-exec}
\setlength{\tabcolsep}{5pt}
\begin{tabular}{lcrrrr}
\toprule
Aggregator & $\widehat{\mathcal V}$ & Bootstrap 95\% & Width
& $100U$ (SE) & $100V$ (SE) \\
\midrule
\multicolumn{6}{l}{\textit{Stage A, logistic regression}} \\
Arithmetic & 0.657 & $[0.626,0.672]$ & 0.046 & $-3.88$ (1.21) & $-7.92$ (2.52) \\
Geometric  & 0.573 & $[0.528,0.589]$ & 0.061 & $-1.83$ (1.05) & $-3.69$ (2.14) \\
RMS & 0.703 & $[0.680,0.716]$ & 0.036 & $-4.70$ (1.20) & $-9.62$ (2.51) \\
TOPSIS & 0.417 & $[0.390,0.431]$ & 0.042 & $-0.41$ (1.14) & $-0.82$ (2.28) \\
\midrule
\multicolumn{6}{l}{\textit{Stage A, random forest}} \\
Arithmetic & 0.652 & $[0.621,0.678]$ & 0.057 & $+1.31$ (1.09) & $+2.59$ (2.16) \\
Geometric & 0.569 & $[0.498,0.601]$ & 0.102 & $-0.38$ (0.72) & $-0.76$ (1.45) \\
RMS & 0.697 & $[0.675,0.719]$ & 0.044 & $+3.98$ (1.08) & $+7.80$ (2.07) \\
TOPSIS & 0.420 & $[0.395,0.441]$ & 0.046 & $-4.30$ (1.13) & $-8.79$ (2.36) \\
\midrule
\multicolumn{6}{l}{\textit{Stage B, logistic regression}} \\
Arithmetic & 0.586 & $[0.557,0.613]$ & 0.056 & $+1.33$ (1.03) & $+2.65$ (2.04) \\
Geometric  & 0.494 & $[0.454,0.525]$ & 0.071 & $+3.91$ (0.90) & $+7.68$ (1.72) \\
RMS & 0.642 & $[0.620,0.665]$ & 0.044 & $-1.63$ (1.16) & $-3.28$ (2.35) \\
TOPSIS & 0.353 & $[0.323,0.381]$ & 0.058 & $+5.42$ (0.91) & $+10.55$ (1.72) \\
\midrule
\multicolumn{6}{l}{\textit{Stage B, random forest}} \\
Arithmetic & 0.632 & $[0.604,0.654]$ & 0.050 & $+4.77$ (1.11) & $+9.31$ (2.11) \\
Geometric & 0.554 & $[0.517,0.577]$ & 0.060 & $+3.50$ (0.94) & $+6.87$ (1.82) \\
RMS & 0.670 & $[0.646,0.691]$ & 0.045 & $+5.14$ (1.16) & $+10.02$ (2.19) \\
TOPSIS & 0.418 & $[0.393,0.438]$ & 0.045 & $+0.35$ (1.13) & $+0.70$ (2.25) \\
\bottomrule
\end{tabular}
\end{table}

Table~\ref{tab:volume-exec} shows that, under tail-swap perturbations, the signed effect of assuming independence ranges from
$-4.7\%$ to $+5.4\%$ for the confidence width and from $-9.6\%$ to $+10.6\%$ for the variance. Thus, depending on the stage, model, and aggregation rule, assuming independence may either overstate or understate uncertainty.

The perturbation family also has a substantial effect on both the composite value and the sign of the independence effect. Replacing tail-swap perturbations with the scaled-Gaussian sweep, evaluated using a fixed draw at each severity level, increases every composite by between $1.1$ and $3.4$ interval widths and makes the width effect uniformly positive, ranging from $+2.5\%$ to $+21.8\%$. The variance effect reaches $38.8\%$ for the Stage B random-forest RMS composite. This change is also reflected in the within-stage correlations. For example, for the Stage A logistic regression, the correlation between the accuracy and robustness curve summaries changes from $-0.02$ under tail swaps to $+0.38$ under Gaussian perturbations.

In summary, the simulated agentic application shows that a reproducible uncertainty certificate should record the fitted model, aggregation rule, covariance structure, and perturbation family.

\section{Conclusions}\label{sec:discussion}

The growing use of Artificial Intelligence requires the development of compliance scores that can support assessments of trustworthiness. These scores should be embedded in reproducible uncertainty certificates that report both the compliance measurement and the conditions under which it was obtained.

In this paper, we have shown how to quantify uncertainty in a compliance score for AI applications, including agentic AI systems, leveraging the SAFE AI rank graduation metrics and the corresponding composite score. Our main contribution is to provide confidence intervals for the composite score and to quantify the effect of assuming independence across different metric components and different stages of a compliance decision. 

We applied our proposal in two different contexts: high-risk AI systems designed to produce credit scores in support of credit-lending decisions, and an agentic AI system designed to make sequential decisions along a predefined pipeline.

Both applications show that a uncertainty certificate should include, besides the compliance score, the correlations among the component metrics or across the different stages. The effect of assuming independence depends on the importance of each component and on the exposure-weighted correlations. It also varies across machine learning models and aggregation rules. A uncertainty certificate should also include the perturbation family employed to assess robustness.

In a nutshell, we have contributed to the development of a statistically supported uncertainty framework for AI compliance certification. Such a framework could be useful for different types of stakeholders, including providers and deployers of AI systems, regulators, and users, in assessing the trustworthiness and reliability of an AI system before it is used.

Future work will consider extending the proposal to full agentic workflows, extending  the SAFE-by-design approach of \citep{babaei2026agentic} from narrow AI to agentic AI.

\bibliographystyle{plainnat}
\bibliography{references}

\section*{Data statement and reproducibility}

The employed data and code are fully reproducible, and are available at:

\url{https://github.com/safe-composed-uncertainty/composed-uncertainty}.
 The original German credit data is available at 
OpenML \texttt{credit-g} v1.

Our considered application runs in 15 minutes on one
workstation, of which 91\% is a conditioning check; the aggregation family
itself costs less than one second once the replicates exist.

\end{document}